\documentclass[11pt]{article}
\usepackage{amsmath,amsthm,amsfonts,amssymb,latexsym,verbatim,bbm}
\usepackage{hyperref}

\input xy
\xyoption{all}

\theoremstyle{plain}
\newtheorem{theorem}{Theorem}[section]
\newtheorem{thm}[theorem]{Theorem}
\newtheorem{lemma}[theorem]{Lemma}

\newtheorem{prop}[theorem]{Proposition}

\newtheorem{defn}[theorem]{Definition}

\newtheorem{corollary}[theorem]{Corollary}
\newtheorem{cor}[theorem]{Corollary}

\newcommand{\R}{\mathbb{R}}
\newcommand{\C}{\mathbb{C}}

\newcommand{\mcA}{\mathcal{A}}
\newcommand{\mcC}{\mathcal{C}}
\newcommand{\mcB}{\mathcal{B}}

\newcommand{\ket}[1]{|{#1}\rangle}
\newcommand{\bra}[1]{\langle {#1} |}
\newcommand{\norm}[1]{\left|\left| {#1}\right|\right|}
\newcommand{\iso}{\cong}

\newcommand{\eps}{\epsilon}
\newcommand{\bias}{\varepsilon}
\DeclareMathOperator{\tr}{tr}
\DeclareMathOperator{\cl}{cl}

\DeclareMathOperator{\vspan}{span}
\DeclareMathOperator{\rank}{rank}
\DeclareMathOperator{\Id}{\mathbbm{1}}

\DeclareMathOperator{\CHSH}{CHSH}
\DeclareMathOperator{\Ball}{B}

\newcommand{\arr}{\rightarrow}

\begin{document}
\title{Lower bounds on the entanglement needed to play XOR non-local games}
\author{William Slofstra%
\date{July 15, 2010}
\thanks{Department of Mathematics, University of California, Berkeley. Email: {\tt slofstra@math.berkeley.edu}}}
\maketitle

\begin{abstract} 
    We give an explicit family of XOR games with $O(n)$-bit questions requiring
    $2^n$ ebits to play near-optimally. More generally we introduce a new
    technique for proving lower bounds on the amount of entanglement required
    by an XOR game: we show that near-optimal strategies for an XOR game $G$
    correspond to approximate representations of a certain $C^*$-algebra
    associated to $G$.  Our results extend an earlier theorem of Tsirelson
    characterising the set of quantum strategies which implement extremal
    quantum correlations. 

\end{abstract}

\section{Introduction}

The purpose of this paper is to study the amount of entanglement required
to play an XOR non-local game optimally or near-optimally. In a non-local game,
Alice and Bob are asked questions chosen at random according to a known
distribution.  They win if their answers match the answers required by the
game's rules, and lose otherwise. They are physically separated and unable to
communicate during the game, so they cannot hope to win every game.  Instead,
they must meet in advance to determine the best strategy, i.e.\ one that
maximizes the probability that they win.  Bell's celebrated theorem
\cite{bell:1964} states that for some games Alice and Bob can increase their
probability of winning if they employ a quantum strategy, in which they make
use of an entangled quantum state. 

From the viewpoint of quantum computation, a quantum strategy is a means of
performing a distributed computation with a higher success probability than
what can be achieved classically with no communication \cite{brassard:2003}
\cite{cleve:2004}.  From this viewpoint it is natural to want to know how much
entanglement is required to implement an optimal quantum strategy. Entanglement
of a pure state can be measured using either the entropy of entanglement (measured
in ebits), or the dimension of the underlying Hilbert space at each site; the
two measures are related by the fact that any state in $\C^d \times \C^d$ has
at most $\log_2 d$ ebits. In physics, the importance of lower bounds stems from
the fact that non-local games provide empirical evidence that a quantum system
is entangled \cite{werner:2001}. Lower bounds on the entanglement required by a
game imply that the game can be used to verify the degree of entanglement in a
quantum system, at least in principle \cite{brunneretal:2008}. 

In this paper we focus on a subset of non-local games generalizing the
well-known Clauser-Horne-Shomony-Holt (CHSH) game \cite{clauser:1969}: Alice
and Bob each receive exactly one question (drawn from finite sets $S$ and $T$
respectively), each question requires a one-bit answer, and the correct
response depends only on the XOR of Alice's and Bob's answers. Non-local games
meeting these restrictions are known as XOR non-local games. They have
attracted interest in part because of the influential results of Tsirelson
\cite{tsirel87} \cite{tsirel93} which guarantee the existence of an optimal
strategy supported on a finite-dimensional Hilbert space, and make it possible
to find an optimal strategy using semidefinite programming methods. Tsirelson
also proved a lower bound on entanglement for quantum strategies by showing
that strategies which implement extreme points of the set of quantum
correlations require dimension exponential in the rank of the correlation
matrix (see Theorem \ref{T:tsirel1}). 

Tsirelson's lower bound can be dualized (see Proposition \ref{P:tsireltight})
to show that there are many XOR games requiring dimension exponential in the
size of the question sets. In particular non-local games exist which are able
to verify that quantum states have arbitrarily high dimension. More recent work
has readdressed this issue with different methods. Junge, Palazuelos,
Perez-Garcia, Villanueva, and Wolf study the ratio $\eps(G,m) / \eps(G,n)$ for
general two and three player non-local games, where $\eps(G,d)$ is the maximum
success bias for $G$ achievable with a quantum system of dimension $d$
\cite{junge:2008} \cite{junge:2009}. In particular, they show that two-player
games exist for which this ratio is arbitrarily high. However, question sets of
size $O(2^{d \log^2 d})$ are required to verify dimension $d$. Briet, Buhrman,
and Toner study a similar ratio for XOR games, where the parameter $d$ in
$\eps(G,d)$ refers to the rank of the quantum correlation matrix rather than
the dimension of the quantum system. The ratio $\eps(G,m) / \eps(G,n)$ cannot
be arbitrarily high for XOR games, since Tsirelson showed that $\eps(G,d) /
\eps(G,1)$ is always bounded by Grothendieck's constant. Nonetheless, Briet,
Buhrman, and Toner show non-constructively that there are XOR games for which
this ratio is greater than $1 + 1/2m + 1/2n - O(1/m^2)$ \cite{briet:2009}.
Vertesi and Pal give an explicit family of XOR games with question sets
of size $O(n^2)$ requiring dimension $n$ to play optimally \cite{vertesi:2008}.
Their lower bound is proved by giving a lower bound on the rank of the optimal
correlation matrix. 

In contrast to the body of work on lower bounds, little is known about upper
bounds on entanglement for general non-local games. Tsirelson proved that XOR
games with question sets of size $m$ and $n$ require at most $\lfloor r / 2
\rfloor$ ebits, where $r$ is the largest integer such that $\binom{r}{2} < m +
n$. One purpose of this paper is to point out that this lower bound is tight,
so XOR games with small question sets can require high entanglement to play
optimally. This can be proved by applying Tsirelson's lower bound method to
specific families of games, such as the family of examples due to Vertesi and
Pal. For completeness, we give another family of examples for which Tsirelson's
bound is tight in Section \ref{S:examples}.  Using a new lower bound technique,
we are also able to study the amount of entanglement required by near-optimal
strategies. Cleve, Hoyer, Toner, and Watrous showed that an $m \times n$ XOR
non-local game has an $\eps$-optimal quantum strategy on a Hilbert space of
dimension $(m+n+1)^{O(\eps^{-2})}$ \cite{cleve:2004}. Their proof can be
modified to be independent of the size of the question sets, so that fixing
$\eps$ bounds the entanglement cost for all XOR games (this was apparently
first observed by O. Regev, and the author is grateful to R. Cleve and R. Jain
for a proof \cite{cleve:private}). We show, for a specific family of games,
that $\eps$-optimal strategies require a Hilbert space of dimension
$\min(O(\eps^{-1/12}), 2^{\lfloor r/2\rfloor})$, where $r$ is the integer from
Tsirelson's upper bound. 

The main idea of this paper is to associate to each XOR non-local game a
finitely presented $C^*$-algebra $\mcA$ with the property that (a) optimal
strategies correspond to representations of $\mcA$, and (b) near-optimal
strategies correspond to approximate representations of $\mcA$. The
presentation for $\mcA$ depends only on the game rules and a small amount of
additional information: the game's marginal biases. Inspiration for this idea
comes from the work of Summers and Werner \cite{summers:1987} and Tsirelson
\cite{tsirel87} \cite{tsirel93}. Summers and Werner show that the Clifford
algebra of rank two satisfies property (a) above for the CHSH game.  The main
statement of Tsirelson's lower bound is that if an XOR non-local game has a
unique optimal correlation matrix $C$ then the algebra generated by the
observables of any optimal non-degenerate quantum strategy is isomorphic to the
Clifford algebra with $\rank C$ generators. Our approach is more general, in
that it applies to all XOR games, and also to near-optimal strategies. 

The paper is organized as follows. In the next section we give an overview of
our results, including the necessary background material and the relation to
Tsirelson's results. Marginal biases are introduced in Section \ref{S:mbias},
and the main result for optimal strategies follows in Section \ref{S:exact}. A
link between near-optimal strategies and approximate representations is
established in Section \ref{S:approxrep}. In Section \ref{S:cliffstable} we
show that the Clifford algebra is stable.  Finally, explicit examples are given
in Section \ref{S:examples}.

\section{Definitions and statement of results}\label{S:background}

\subsection{XOR non-local games and quantum strategies}\label{S:qstrat}

The rules for an XOR non-local game are comprised of two finite sets 
$S$ and $T$ of questions, a probability distribution $\pi$ on $S \times T$,
and a function $f : S \times T \arr \{0,1\}$ recording the correct answers.
Alice and Bob have access to the rules and may communicate before the
game begins. When the game begins, a pair of questions $(i,j)$ is chosen
with probability $\pi(i,j)$. Alice receives question $i$, and must output
a bit $a$; Bob receives question $j$ and outputs a bit $b$. Alice and Bob
win if $a \oplus b$ (the XOR of the two bits) matches $f(i,j)$. All that
matters about $S$ and $T$ is their size, so we assume $S = \{1,\ldots,
m\}$ and $T = \{1,\ldots,n\}$. The cost matrix for this game is defined to be
the $m \times n$ matrix $G$ with $G_{st} = (-1)^{f(i,j)} \pi(i,j)$. The value
of $\pi(i,j)$ can always be recovered from $G_{ij}$, and $f(i,j)$ can also be
recovered if $\pi(i,j) \neq 0$.  If $\pi(i,j) = 0$ then $f(i,j)$ is irrelevant
to the game, so an $m \times n$ XOR non-local game is completely described
by its cost matrix $G$. A matrix $G$ is the cost matrix of an XOR non-local
game if $\sum |G_{ij}| = 1$.

Let $a_i$ be the binary random variable corresponding to Alice's output on
input $i$, and let $b_j$ be the binary random variable corresponding to
Bob's output on input $j$. Then $(-1)^{f(i,j)} (-1)^{a_i} (-1)^{b_j}$ is $1$ if
Alice and Bob win on question $i,j$, and $-1$ otherwise. The sum
\begin{equation*}
    \sum_{i,j} G_{ij} (-1)^{a_i} (-1)^{b_j}
\end{equation*}
is the \emph{expected success bias} for this behaviour, analagous to the
expected success probability, but normalized between $-1$ and $1$.\footnote{If
$p$ is the success probability, then the success bias is $2p-1$.} In the
quantum setting, Alice and Bob are allowed to select a shared bipartite quantum
state $\ket{\psi} \in H_1 \otimes H_2$. In addition to the state $\ket{\psi}$,
a quantum strategy consists of two families of observables $\{A_i : i \in S\}$
and $\{B_j : j \in T\}$, on the Hilbert spaces $H_1$ and $H_2$ respectively,
with spectra contained in $\{-1,1\}$. The success bias for this quantum
strategy is
\begin{equation*}
    \sum G_{ij} \bra{\psi} A_i \otimes B_j \ket{\psi},
\end{equation*}
and the \emph{quantum success bias $\bias_q(G)$} for the game is the maximum
success bias across all quantum strategies. The amount of entanglement used by
a quantum strategy will be measured using the entropy of entanglement of the
state $\ket{\psi}$, which is by definition the von Neumann entropy of the
partial trace of $\ket{\psi}$ with respect to $H_1$. The dimension of a quantum
strategy is $\min (\dim H_1, \dim H_2)$. 

\subsection{$C^*$ algebras generated by self-adjoint indeterminates} 

Let $R_1,\ldots,R_k$ be non-commutative polynomials in indeterminates
$X_1,\ldots,X_n$. Let $\mcA$ be a $C^\ast$-algebra generated by self-adjoint
operators $X_1,\ldots,X_n$. Then $\mcA$ is said to be the universal
$C^\ast$-algebra generated by self-adjoint indeterminates $X_i$  satisfying
relations $R_1,\ldots,R_k$ if for all Hilbert spaces $H$ the map $\rho \mapsto
(\rho(X_1),\ldots,\rho(X_n))$ is a bijection between the representations $\rho$
of $\mcA$ on $H$ and tuples $(B_1,\ldots,B_n)$ of bounded self-adjoint
operators on $H$ such that $R_i(B_1,\ldots,B_n) = 0$ for all $i=1,\ldots,k$.
The universal $C^\ast$-algebra for a set of relations is determined uniquely up
to isomorphism, if it exists. The set of relations determining $\mcA$ is not
unique. Since we are interested in $C^\ast$-algebras given by certain
relations, we will regard the generators and defining relations as part of the
data of the algebra. A universal $C^\ast$-algebra for a given set of relations
exists if and only if there are constants $b_1,\ldots,b_n$ such that
$\norm{B_i} \leq b_i$ whenever $(B_1,\ldots,B_n)$ is a tuple of operators
satisfying the given relations. For convenience we introduce a slight
generalization of the usual notion of a cyclic representation:

\begin{defn}
    Let $\mcA$ be a $C^*$-algebra generated by self-adjoint indeterminates
    $X_i$ satisfying relations $R_j$. A \emph{density-matrix representation} of
    $\mcA$ is an action of $\mcA$ on a Hilbert space $H$, determined by a
    collection of self-adjoint operators $B_i$, along with a density matrix
    $\rho$ on $H$ such that $\mcA \rho H$ is dense in $H$. 
\end{defn}

The following universal $C^*$-algebra will be used in examples, and in 
describing previous results of Tsirelson.
\begin{defn}
    The Clifford algebra of rank $r$ is the universal $C^\ast$-algebra
    $\mcC_r$ generated by indeterminates $X_1,\ldots,X_r$ satisfying the
    relations $X_i^2 = \Id$ for all $1 \leq i \leq n$ and $X_i X_j = 
    - X_j X_i$ for all $i \neq j$. 
\end{defn}
A universal $C^\ast$-algebra generated by $Y_1,\ldots,Y_n$ is \emph{Clifford}
if for all $1 \leq i,j \leq n$ the anti-commutator $Y_i Y_j + Y_j Y_i$ is a
scalar multiple of the identity.  If $\mcA$ is Clifford with $Y_i Y_j + Y_j Y_i
= C_{ij} \Id$ then $\mcA \iso \mcC_r$, where $r$ is the rank of the matrix
$\{C_{ij}\}$. The isomorphism can be chosen so that every $Y_i$ corresponds
to a linear combination of the distinguished generators $X_j$ of $C_r$. We
will call a universal $C^\ast$-algebra \emph{strongly Clifford} if the
Clifford relations $Y_i Y_j + Y_j Y_i - V_{ij} \Id$ are linear combinations
of the defining relations.  The representation theory of $C_r$ is well-known:
$C_r$ has either one or two irreducible representations of dimension
$2^{\lfloor r / 2 \rfloor}$.

\subsection{Main result: algebraic characterization of optimal
solutions}\label{S:cstar}

If $\{A_i\}$,$\{B_j\}$,$\ket{\psi}$ is a quantum strategy on a bipartite
Hilbert space $H_1 \otimes H_2$, define the \emph{marginal (strategy)} on $H_2$
to be the collection of operators $\{B_j\}$, as well as the density operator
$\rho$ which is the partial trace of $\ket{\psi}$ with respect to $H_1$. Note
that the entanglement of the original strategy is the von Neumann entropy of
$\rho$. A strategy $\{A_i\}$,$\{B_j\}$,$\ket{\psi}$ is said to be
non-degenerate if there is no projection $P$ commuting with all $A_i$ such that
that $(P \otimes I) \ket{\psi} = \ket{\psi}$, and no projection $Q$ commuting
with all $B_j$ such that $(I \otimes Q) \ket{\psi} = \ket{\psi}$. Every
degenerate strategy projects down to a unique non-degenerate strategy. 

Our main result is a precise description of the marginal strategies
corresponding to optimal non-degenerate quantum strategies of a given XOR
non-local game $G$. A key part of this description is the fact that for every
$m \times n$ XOR non-local game $G$ and $1 \leq i \leq m$ there is a constant
$c_i$ such that $\sum_j G_{ij} \bra{\psi} A_i \otimes B_j \ket{\psi} = c_i$
whenever $\{A_i\}$,$\{B_j\}$,$\ket{\psi}$ is an optimal quantum strategy. We
refer to $c_i$ as the \emph{$i$th marginal row bias for $G$}. The existence of
the marginal biases is proved in Section \ref{S:mbias}. 

If $G$ is an $m \times n$ XOR non-local game with marginal row biases $c_i$,
define the \emph{solution algebra $\mcA$ for $G$} to be the universal $C^*$
algebra generated by self-adjoint indeterminates $X_1,\ldots,X_n$, satisfying
the relations
\begin{equation*}
    X_j^2 = \Id \text{ for all } 1 \leq j \leq n, \text{ and }
    \left( \sum_j G_{ij} X_j \right)^2 = c_i^2 \cdot \Id 
        \text{ for all } 1 \leq i \leq m.
\end{equation*}
\begin{thm}\label{T:tensordescript} Let $G$ be a XOR non-local game with no zero
    rows or columns and solution algebra $\mcA$. A collection of bounded linear
    operators $\{B_j : j \in T\}$ and density operator $\rho$ on a Hilbert
    space $H$ is the marginal of a non-degenerate optimal strategy for $G$ if
    and only if the map $X_j \mapsto B_j$ induces a density-matrix
    representation of $\mcA$ on $H$ with density matrix $\rho$, and $\rho$
    commutes with the image of $\mcA$.
\end{thm}
Theorem \ref{T:tensordescript} determines the entanglement required by $G$ in
the following sense:
\begin{cor}\label{C:minentangl} Let $G$ be an XOR non-local game with no zero
    rows or columns, and let $\mcA$ be the corresponding solution algebra. Let
    $N$ be the minimum dimension among non-zero representations of $\mcA$. Then
    the minimum entanglement used by an optimal quantum strategy for $G$ is
    $\log_2 N$. The strategies which attain the minimum entanglement are the
    irreducible representations of $\mcA$ of dimension $N$, with cyclic state
    $I/N$, where $I$ denotes the identity operator.
\end{cor}
The dimension of a representation is the dimension of the underlying Hilbert
space. Although the results above are stated for pure states, Corollary
\ref{C:minentangl} extends trivially to mixed states when the entanglement of
formation is used as an entanglement measure. A mixed state of minimum
entanglement must be a mixture of maximally entangled states.  The proofs of
Theorem \ref{T:tensordescript} and Corollary \ref{C:minentangl} are given in
Section \ref{S:exact}.

\subsection{Comparison to Tsirelson's results}

An $m \times n$ matrix $\{c_{ij}\}$ is called a \emph{quantum correlation} if
there is a quantum strategy $\{A_i\}$,$\{B_j\}$,$\ket{\psi}$ such that $c_{ij}
= \bra{\psi} A_i \otimes B_j \ket{\psi}$ for all $i,j$. The set $C_{mn}$ of all
$m \times n$ quantum correlations is closed and convex, so the problem of
finding the quantum success bias for a game $G$ can be formulated as the
following convex programming problem:
\begin{equation*}
    (\Phi) \quad\quad \max\ \sum_{i,j} G_{ij} c_{ij} \ :\ \text{ where }
        \{c_{ij}\} \text{ is a quantum correlation.}
\end{equation*}
The optimisation problem $(\Phi)$ can be formulated as a semidefinite
programming problem using the following theorem. This makes it possible
to find a near optimal quantum correlation in practice. 
\begin{thm}[Tsirelson, 1987]\label{T:correleq} An $m \times n$ matrix
    $\{c_{ij}\}$ is a quantum correlation matrix if and only if there are two
    families of vectors $\{u_i \in \Ball\left(\R^N\right) : i \in S\}$ and
    $\{v_j \in \Ball\left(\R^N\right) : j \in T\}$ such that $c_{ij} = u_i
    \cdot v_j$. 
\end{thm}
Here $\Ball\left(\R^N\right)$ denotes the unit ball in $\R^N$. Define a
\emph{vector strategy} for $\{c_{ij}\}$ to be two collections of vectors
$\{u_i\}$ and $\{v_j\}$ in $\Ball\left(\R^N\right)$ such that $c_{ij} = u_i
\cdot v_j$. The proof of Theorem \ref{T:correleq} is constructive, associating
to a vector strategy a quantum strategy $\{A_i\}$, $\{B_j\}$, $\ket{\psi}$ with
the same correlation. The observables $\{A_i\}$ and $\{B_j\}$ are constructed
using the Clifford algebra $C_N$, and consequently Theorem \ref{T:correleq}
implies that every $m \times n$ game has an optimal quantum strategy using
$\lfloor n/2\rfloor$ ebits. Tsirelson \cite{tsirel87} showed that this upper
bound could be improved to $\lfloor r/2 \rfloor$ ebits, where $r$ is the
largest integer such that $\binom{r+1}{2} < m+n$. 
\begin{prop}\label{P:tsireltight}
    For every $n \geq 2$ there is an $m \times n$ XOR game for which
    Tsirelson's bound is tight. 
\end{prop}
\begin{proof}
    Theorem 2.22 of \cite{tsirel93} can be used to find a maximal rank
    extreme point of the set of quantum correlations. Any supporting
    hyperplane for this extreme point will satisfy the proposition.

    For a constructive proof, apply either Theorem \ref{T:tensordescript} or
    Theorem \ref{T:tsirel1} below to explicit examples, such as the family of
    games in Section \ref{S:examples} ($m = \binom{n}{2}$) or the similar
    family of games in \cite{vertesi:2008} ($m = \binom{n}{2} + 1$).
\end{proof}

There is a strong relation between Theorem \ref{T:tensordescript} and the
following theorem of Tsirelson.
\begin{thm}[Tsirelson \cite{tsirel87}]\label{T:tsirel1} Suppose that $\{c_{ij}\}$
    is an extreme point of $C_{mn}$ of rank $r$, and $\{A_i\}$, $\{B_j\}$,
    $\ket{\psi}$ is a non-degenerate strategy\footnote{Tsirelson actually
    proves this result for a slightly different definition of quantum
    strategy, but the result still holds for the definition in use here.}
    representing $\{c_{ij}\}$. Then the strategy $\{A_i\}$, $\{B_j\}$,
    $\ket{\psi}$ is Clifford, and $\ket{\psi}$ uses at least $\lfloor r/2
    \rfloor$ ebits.  
\end{thm}
One way of understanding Tsirelson's theorem is that it associates a Clifford
algebra to every extreme point of $C_{mn}$. XOR non-local games are in
one-to-one correspondence with supporting hyperplanes of $C_{mn}$ via the cost
matrix, so Theorem \ref{T:tensordescript} describes how to associate a
$C^*$-algebra to every supporting hyperplane of $C_{mn}$. Thus the two results
are complementary, and correspond to two approaches to describing convex sets:
by extreme points and by supporting hyperplanes. Furthermore, we now have two
ways of showing that an XOR non-local game requires high entanglement: if the
representations of the solution algebra are known then we can apply Corollary
\ref{C:minentangl}; and if we can show that the optimization problem $(\Phi)$
has a unique solution of high rank then this solution must be an extreme point
of $C_{mn}$ and we can apply Theorem \ref{T:tsirel1}.  Both approaches work for
the examples of games requiring high entanglement presented in Section
\ref{S:examples}. 

This raises the question of whether Theorems \ref{T:tensordescript} and
\ref{C:minentangl} tell us anything new about optimal strategies and
entanglement requirements, especially since $\mcA$ is almost always Clifford.
The answer is yes, for interesting reasons. Theorem \ref{T:tsirel1} does
not give a full description of the optimal strategies for a game $G$ unless
$(\Phi)$ has a unique solution.
\begin{prop}\label{P:cliffuniq} The solution algebra $\mcA$ for a game $G$ with
    no zero rows or columns is Clifford if and only if there is a unique
    quantum correlation which is optimal for $G$.
\end{prop}
\begin{proof}
    For $\mcA$ to be Clifford means that $X_i X_j + X_j X_i = 2 V_{ij} \Id$ for
    all $i,j$ and some fixed $n \times n$ matrix $V$. In Proposition
    \ref{C:nonlinear}, we will give an optimization problem $(\Gamma)$ with the
    property that $X_i X_j + X_j X_i = 2 V_{ij} \Id$ in some representation of
    $\mcA$ if and only if $V$ is an optimal solution to $(\Gamma)$. When $G$
    has no zero rows or columns there is a one-to-one relationship between
    optimal solutions of $(\Gamma)$ and optimal quantum correlations. If $\mcA$
    is Clifford then $(\Gamma)$ will have a unique solution, so there must be a
    unique optimal quantum correlation.

    On the other hand, if there is a unique quantum correlation, then
    $(\Gamma)$ will also have a unique solution. By Theorem \ref{T:tsirel1},
    the relation $X_i X_j + X_j X_i = 2 V_{ij} \Id$ will hold in every
    representation of $\mcA$, with $V$ the unique solution to $(\Gamma)$, so
    this relation holds in $\mcA$ itself.
\end{proof}
In addition, the strategy of minimum entanglement is not always Clifford. 
\begin{prop}\label{P:goodexam}
    For any $n \geq 1$ there is a XOR non-local game such that any optimal
    strategy which is Clifford uses at least $n$ ebits more than the strategy
    of minimum entanglement. For an example, see the game $CL(n)$ described in
    Section \ref{S:examples}.
\end{prop}
In particular, the correlation of a strategy of minimum entanglement does not
have to be an extreme point of $C_{mn}$, and the minimum entanglement required
by a game is not a function of the rank of the optimal quantum correlation. 

\subsection{Stability and approximate strategies}

Stability questions arise naturally in applications of Theorem
\ref{T:tensordescript}. For example, suppose we only have an
approximation $\hat{c}_i$ to the marginal biases $c_i$. Are the
representations of the algebra determined by the relations
\begin{equation*}
    X_j^2 = \Id \text{ for all } 1 \leq j \leq n, \text{ and }
    \left( \sum_j G_{ij} X_j \right)^2 = \hat{c}_i^2 \cdot \Id
        \text{ for all } 1 \leq i \leq m.
\end{equation*}
close to representations of the solution algebra $\mcA$? This question is
related to the study of stable relations for $C^\ast$-algebras, and the broader
area of Hyers-Rasiass-Ulam stability in functional analysis and group theory.
We use the following definition of stability due to \cite{loring:1993}---note
that a norm condition $\norm{B_i} \leq \norm{X_i}_{\mcA}$ is added to simplify
some of the later analysis. 
\begin{defn} Let $\mcA$ be a universal $C^\ast$-algebra generated by
    self-adjoint indeterminates $X_1,\ldots,X_n$ satisfying relations
    $R_1,\ldots,R_k$.  An $\eps$-approximate representation of $\mcA$ is a
    collection of bounded self-adjoint operators $B_1,\ldots,B_n$ on a Hilbert
    space $H$, such that $\norm{B_i} \leq \norm{X_i}_{\mcA}$ for all $1 \leq i
    \leq n$ and $\norm{R_i(B_1,\ldots,B_n)} \leq \eps$ for all $1 \leq i \leq
    k$.

    A set of relations $R_1,\ldots,R_k$ is \emph{stable} (for
    finite-dimensional Hilbert spaces) if there is a constant $\delta > 0$ and
    function $f : [0,\delta) \arr [0,+\infty)$ such that if $0 < \eps < \delta$
    and $(B_1,\ldots,B_n)$ is an $\eps$-approximate represention of $\mcA$ on a
    finite-dimensional Hilbert space $H$ then there is a representation
    $B_1',\ldots,B_n'$ of $\mcA$ on $H$ with $\norm{B_i - B_i'} \leq f(\eps)$.
\end{defn}
A note on norms: we use $\norm{\cdot}$ to denote the operator
norm induced by the Hilbert space structure. On finite-dimensional Hilbert
spaces we can also use the Frobenius norm $\norm{X}_F = \sqrt{\tr(X^\ast X)}$.
The two norms are related by the inequality $\norm{\cdot} \leq \norm{\cdot}_F
\leq \sqrt{d} \norm{\cdot}$, where $d$ is the dimension. 

\begin{prop}\label{P:cliffstable} 
    \
    \begin{itemize}
        \item The Clifford algebra $C_r$ is stable with constant $\delta = 1 /
            (250 r^2)$ and function $f(\eps) = 5 r^2 \eps / 2$.
        \item If $\mcA$ is strongly Clifford of rank $r$ then there is a
            constant $\delta > 0$, such that if $\mcA$ has an $\eps$-approximate
            representation with $\eps < \delta$ then there is an
            exact representation of $\mcA$ of the same dimension.
        \item Almost all $m \times n$ games with $m \geq \binom{n}{2}$
            are strongly Clifford.
    \end{itemize}
\end{prop}
The proof of Proposition \ref{P:cliffstable} is given in Section
\ref{S:cliffstable}. Although we do not make use of the fact, the solution
algebra of an $m \times n$ game is stable if it is strongly Clifford of rank
$n$.  The author does not know whether or not the relations for the solution
algebra are stable in general. 

A quantum strategy is said to be $\eps$-optimal for a game $G$ if it has
success bias within $\eps$ of $\bias_q(G)$. It is possible to use the solution
algebra to study $\eps$-optimal strategies as $\eps \arr 0$.  It will be clear
from the proof of Theorem \ref{T:tensordescript} that an approximate
representation of $\mcA$ will be the marginal of a near-optimal strategy for
$G$. A weak converse is also true. 
\begin{thm}\label{T:approxrep}
    For every $m \times n$ XOR-non-local game with no zero rows or columns
    there are constants $C,C' > 0$ such that if $0 < \eps < C$ and $\{B_j : j
    \in T\}$,$\rho$ is the marginal of an $\eps$-optimal strategy supported on
    a Hilbert space $H$ with $\dim H = d < +\infty$, then there is a projection
    $P$ on $H$ such that $\{P B_j P\}$ is a $(C' d^{3/2}
    \eps^{1/8})$-approximate representation for the solution algebra of $G$.
\end{thm}
Specifically, $C \geq \min_j d_j^{16} / (100 (m+n))$ ($d_j$ is the $j$th
marginal column bias), and $C' \leq 15 (m+n)^{\frac{1}{8}}$. The proof of this
theorem is given in Section \ref{S:approxrep}.

For near-optimal strategies we would like to understand the dimensions of the
Hilbert spaces supporting $\eps$-optimal strategies, for $\eps$ in a
neighbourhood of zero. For this to make sense, we look at finite-dimensional
spaces only. Cleve, Hoyer, Toner, and Watrous have shown that an $m \times n$
XOR non-local game has an $\eps$-optimal quantum strategy on a Hilbert space of
dimension $(m+n+1)^{O(\eps^{-2})}$ \cite{cleve:2004}. Regev has pointed out
that this bound can be made independent of $m$ and $n$, so that fixing the
precision $\eps$ upper bounds the entanglement cost for all XOR non-local games
\cite{cleve:private}.  Theorem \ref{T:approxrep} gives lower bounds on the
dimension of approximate solutions when the solution algebra $\mcA$ is strongly
Clifford.

\begin{corollary}\label{C:dimbound} Let $G$ be an XOR non-local game with no
    zero rows or columns, and solution algebra $\mcA$ which is strongly
    Clifford of rank $r$.  Then there are constants $C$ and $C''$ such that if
    $0 < \eps \leq C$ then any $\eps$-optimal strategy is supported on a
    Hilbert space of dimension at least $\min \left( C'' \eps^{-1/12},
    2^{\lfloor r/2 \rfloor} \right)$.
\end{corollary}
\begin{proof}
    Let $C,C'$ to be the constants from Theorem \ref{T:approxrep}. Suppose $G$
    has an $\eps$-optimal solution on a Hilbert space $H$ of dimension $d$, $0
    \leq \eps \leq C$. By Theorem \ref{T:approxrep}, $\mcA$ has a $C' d^{3/2}
    \eps^{1/8}$-approximate representation of dimension $d$.  Take $C'' = (C' /
    \delta)^{-2/3}$, where $\delta$ is the constant from Proposition
    \ref{P:cliffstable}. If $d < C'' \eps^{-1/12}$ then 
    \begin{equation*} 
        C' d^{3/2} \eps^{1/8} < \delta,
    \end{equation*} 
    so Proposition \ref{P:cliffstable} implies that $\mcC_r$ has an exact
    representation on $H$, implying that $d \geq 2^{\lfloor r/2 \rfloor}$.
\end{proof}
The lower bound in Corollary \ref{C:dimbound} appears to be quite weak.
For instance, it is possible to show that dimension increases linearly with
$1-\eps$ up to dimension $n$, for a specific family of games with question sets
of size $n(n-1)$ and $n$ (see Proposition \ref{P:vectlb}).

\section{Marginal biases}\label{S:mbias}

\begin{thm}\label{T:mbias}
    For every $m \times n$ XOR non-local game $G$ there is a collection of
    non-negative constants $\{ c_i\ :\ i \in S\}$ such that if a vector strategy
    $\{u_i\}$,$\{v_j\}$ for $G$ is $\eps$-optimal, $0 \leq \eps <
    1/\left(4 (m+n)\right)$, then
    \begin{equation*}
        \norm{\sum_j G_{ij} v_j - c_i u_i} \leq \sqrt{10} (m+n)^{1/4} \eps^{1/4}.
    \end{equation*}
    If $G$ has no zero rows then the constants $c_i$ are non-zero.
\end{thm}
When $\eps = 0$, $\sum_j G_{ij} v_j = c_i u_i$, so the sums $\sum_g G_{ij}
u_i \cdot v_j$ are constant (and equal to $c_i$) for every optimal vector
strategy $\{u_i\}$,$\{v_j\}$. There are also \emph{marginal column biases}
$d_j$ playing the same role for the columns. The marginal column biases are
non-zero when $G$ has no zero columns.  

\begin{cor}\label{C:mbias} Let $G$ be an XOR non-local game with no zero rows.
    A vector strategy $\{u_i\}$,$\{v_j\}$ for $G$ is optimal if and only if
    \begin{equation*}
        \sum_j G_{ij} v_j = c_i u_i \text{ for all } 1 \leq i \leq m,
    \end{equation*}
    where $c_i$ is the $i$th marginal row bias of $G$.
\end{cor}
Alternatively, we can characterize optimal vector strategies without any
reference to the vectors $u_i$ at all.
\begin{cor}\label{C:nonlinear} Let $\{v_j\}$ be a collection of $n$ unit vectors.
    Then there are unit vectors $\{u_i\}$ such that the quantum correlation $\{
    u_i \cdot v_j\}$ is an optimal solution to $G$ if and only if the positive
    semidefinite matrix $\{v_i \cdot v_j\}$ is an optimal solution to 
    \begin{align*}
        (\Gamma) \quad\quad \max\ \sum_i \ \sqrt{\sum_{j,k} G_{ij}G_{ik}V_{jk}} 
            \ :\ V\text{ is an } n\times n \text{ positive semidefinite matrix 
                    with } V_{ii} = 1.
    \end{align*}
    If $G$ has no zero rows or columns there is a one-to-one relationship
    between solutions to $(\Gamma)$ and optimal quantum correlations for $G$.
    If $V$ is an optimal solution then the $i$th marginal row bias is given
    by $c_i^2 = \sum_{j,k} G_{ij}G_{ik}V_{jk}$.
\end{cor}

To prove Theorem \ref{T:mbias}, consider the semidefinite programming (SDP)
formulation for optimal value of $G$ given in \cite{cleve:2007}. Specifically, let 
\begin{equation*}
    B = \begin{pmatrix} 0 & \frac{1}{2} G \\ \frac{1}{2} G^T & 0 \end{pmatrix}.
\end{equation*}
Then $\bias_q(G)$ is equal to
\begin{equation*}
    (P) \quad \quad \max\ \tr(B X)\ :\ X_{ii} = 1 \text{ for } 1 \leq i \leq m+n, 
        \text{ and } X \succeq 0,
\end{equation*}
where $\succeq$ denotes the partial order induced by the semidefinite cone. An
overview of semidefinite programming techniques can be found in
\cite{boyd:1996}. There is a direct correspondence between vector strategies
and feasible solutions of $(P)$: if $\{u_i\}$,$\{v_j\}$ is a vector strategy,
$U$ is the matrix with columns given by the vectors $u_i$, and $V$ is the matrix
with columns given by the vectors $v_j$, then the corresponding feasible solution to
$(P)$ is the matrix
\begin{equation}\label{yousandvees} 
    X = \begin{pmatrix} U^T \\ V^T \end{pmatrix} \cdot 
        \begin{pmatrix} U & V \end{pmatrix} 
        = \begin{pmatrix} U^T U & U^T V \\ V^T U & V^T V \end{pmatrix},
\end{equation}
and this construction can be run in reverse as well. The success bias
can also be calculated using the dual formulation:
\begin{align*}
    (D) \quad\quad \min \frac{1}{2} \sum_{1 \leq i \leq m} c_i + \frac{1}{2} \sum_{1 \leq i \leq n} d_j\ : \\
        \frac{1}{2} \begin{pmatrix} \Delta(c) & 0 \\ 0 & \Delta(d)\end{pmatrix} \succeq B.
\end{align*}
Here $\Delta(c)$ denotes the diagonal matrix with the entries of $c$ on the diagonal.
 
\begin{proof}[Proof of Theorem \ref{T:mbias}] 
    Let $c_1,\ldots,c_m$,$d_1,\ldots,d_n$ be optimal for $(D)$, and let 
    \begin{equation*}
        S = \frac{1}{2} \begin{pmatrix} \Delta(c) & 0 \\ 0 & \Delta(d)\end{pmatrix} - B.
    \end{equation*}
    Let $X$ be the positive semidefinite matrix corresponding to a (not
    necessarily optimal) vector strategy $\{u_i\}$, $\{v_j\}$. Then
    \begin{equation*}
        (S X)_{ii} = \begin{cases} 
            \frac{1}{2} c_i - \frac{1}{2} \sum_j G_{ij} u_i \cdot v_j & 1 \leq i \leq m \\
            \frac{1}{2} d_{i-m} - \frac{1}{2} \sum_j G_{j(i-m)} u_j \cdot v_{i-m} & m+1 \leq i \leq m+n \\
        \end{cases},
    \end{equation*}
    and in particular
    \begin{equation*}
        \tr (S X) = \frac{1}{2} \sum_{1 \leq i \leq m} c_i + \frac{1}{2}
            \sum_{1 \leq j \leq n} d_j - \tr (B X) = \bias_q(G) - \tr(B X),
    \end{equation*}
    where the last equality holds by strong duality.  The vector strategy in
    question is $\eps$-optimal if and only if $X$ is $\eps$-optimal for $(P)$,
    meaning that $\bias_q(G) - \tr(B X) \leq \eps$.  The following inequality
    holds for every $1 \leq i \leq m$ and every $\eps$-optimal vector strategy:
    \begin{align*}
        \left| c_i - \sum_j G_{ij} u_i \cdot v_j\right| & = 2 |(S X)_{ii}| 
            \leq 2 \norm{SX}_F \leq 2 \norm{X^{1/2}}_F \norm{S^{1/2}}_F \norm{X^{1/2}
            S^{1/2}}_F \\ 
        & = 2 \sqrt{\tr(X)} \sqrt{\tr(S)} \sqrt{\tr(XS)} \leq 2 \sqrt{m+n}
            \sqrt{\bias_q(G)} \sqrt{\eps},
    \end{align*}
    where $\norm{\cdot}_F$ is the Frobenius norm.  Let $\delta = 2 \sqrt{m+n}
    \sqrt{\bias_q(G)} \sqrt{\eps}$. Then
    \begin{equation*}
        \norm{c_i u_i - \sum_j G_{ij} v_j}^2 = c_i^2 - 2 c_i \sum_j G_{ij} u_i
            \cdot v_j + \norm{ \sum_j G_{ij} v_j}^2.
    \end{equation*}
    We know that $\sum_j G_{ij} u_i \cdot v_j > c_i - \delta$. To bound $\norm{\sum_j G_{ij} v_j}$,
    let $u'_i$ be the vector $\sum_j G_{ij} v_j$, renormalized to a unit vector. Then
    \begin{equation*}
        \sum_j G_{ij} u_i' \cdot v_j = \norm{\sum_j G_{ij} v_j} \geq \sum_{j} G_{ij} u_i \cdot v_j,
    \end{equation*}
    so $\{u_i'\}$,$\{v_j\}$ is also an $\eps$-optimal strategy for $G$. Hence $\norm{\sum_j G_{ij}
    v_j}  = \sum_j G_{ij} u_i' \cdot v_j \leq c_i + \delta$. We can conclude that
    \begin{equation*}
        \norm{c_i u_i - \sum_j G_{ij} v_j}^2 \leq c_i^2 - 2 c_i (c_i - \delta) + (c_i + \delta)^2
            = 4 c_i \delta + \delta^2, 
    \end{equation*}
    and this last expression is at most $5 \bias_q(G)^{1/2} \delta$ if $\eps
    \leq 1/(4(m+n))$. Since $\eps_q(G) \leq 1$, this shows that the constants
    $c_i$ satisfy the statement of Theorem \ref{T:mbias}. Uniqueness follows
    from the fact that $c_i = \sum_j G_{ij} u_i \cdot v_j$ for every optimal
    vector strategy $\{u_i\}$,$\{v_j\}$. Finally $S$ is positive semidefinite,
    so if the $i$th row of $G$ is non-zero then $c_i$ must be strictly
    positive.
\end{proof}

\section{Quantum strategies: the exact case}\label{S:exact}

We begin with the proof of Theorem \ref{T:tensordescript}. For this we
introduce some notation, which we will keep fixed for this and also the next
section.  Suppose $\ket{\psi}$ is an element of $H_1 \otimes H_2$. If we pick a
basis $\ket{i}$ for $H_1$ then we can write $\ket{\psi} = \sum \ket{i} \lambda
\ket{i}$, where $\lambda : H_1 \arr H_2$ is a linear transformation. With this
notation, the Schmidt coefficients for $\ket{\psi}$ are the eigenvalues of
$\lambda$, and the partial trace of $\ket{\psi}$ with respect to $H_1$ is $\rho
= \lambda \lambda^\ast$, a density operator on $H_2$. Having picked a basis on
$H_1$, we can write any element of $\mcB(H_1)$ as a matrix. If $A \in
\mcB(H_1)$, let $\bar{A}$ denote the linear transformation constructed by
taking the entry-wise complex conjugate of the matrix of $A$.
Note that the operation $A \mapsto \bar{A}$ depends on the choice of basis for
$H_1$. 
\begin{lemma}\label{L:tensorcommute} Let $A$ and $B$ be Hermitian operators on
    Hilbert spaces $H_1$ and $H_2$ respectively. Then $\norm{ \left(A \otimes
    \Id - \Id \otimes B \right) \ket{\psi} } = \norm{ \lambda \bar{A} - B
    \lambda }_F.$ Consequently, if $\norm{ \left(A \otimes \Id - \Id \otimes B
    \right) \ket{\psi} } \leq \eps$, then $\norm{ \rho B - B \rho }_F \leq 2
    \eps$, or in other words $\rho$ approximately commutes with $B$.
\end{lemma}
\begin{proof}
    For the first identity, observe that $(A \otimes \Id) \ket{\psi} =
    \sum_i \ket{i} \lambda \bar{A} \ket{i}$, while $(\Id \otimes B) \ket{\psi}
    = \sum_i \ket{i} B \lambda \ket{i}$. So $\norm{\left(A \otimes \Id - \Id
    \otimes B \right) \ket{\psi} }^2 = \norm{\lambda \bar{A} - B \lambda}_F^2$.
   
    If $\norm{ \left(A \otimes \Id - \Id \otimes B \right) \ket{\psi} } \leq
    \eps$ then $\norm{ B \rho  - \lambda \bar{A} \lambda^\ast }_F \leq \norm{ B
    \lambda - \lambda \bar{A} }_F \norm{ \lambda^\ast}_F  \leq \eps$.  But
    $\lambda \bar{A} \lambda^\ast$ is Hermitian, while $(B \rho)^\ast = \rho B$.
\end{proof}

Now we should say something about what it means for an operator strategy to be
non-degenerate. Suppose that $\{A_i\}$, $\{B_j\}$, $\ket{\psi}$ is a quantum
strategy, and let $\mcB_1$ and $\mcB_2$ denote the $C^\ast$-algebras generated by
$A_1,\ldots,A_m$ and $B_1,\ldots,B_n$ respectively. Let $\cl$ denote the
closure of a set in a Hilbert space with respect to the norm topology. 
\begin{lemma}\label{L:nondegen}
    The strategy in question is non-degenerate if and only if $\cl \mcB_2
    \lambda H_1 = H_2$ and $\cl \bar{\mcB}_1 \lambda^\ast H_2 = H_1$.
\end{lemma}
\begin{proof}
    Let $Q$ be a projection onto a closed subspace $W$ of $H_2$. By Lemma
    \ref{L:tensorcommute}, $(I \otimes Q) \ket{\psi} = \ket{\psi}$ if and only
    if $Q \lambda = \lambda$, and this occurs if and only if the image of
    $\lambda$ is contained in $W$. Also $Q$ will commute with $\mcB_2$ if and
    only if $W$ is an invariant subspace for $\mcB_2$. Thus there is a
    projection $Q$ commuting with $\mcB_2$ such that $(I \otimes Q) \ket{\psi}
    = \ket{\psi}$ if and only if $\cl \mcB_2 \lambda H_1$ is a strict subset of
    $H_2$. 

    There is a projection $P$ such that $(P \otimes I) \ket{\psi} = \ket{\psi}$
    if and only if $\lambda \bar{P} = \lambda$, and this is equivalent to 
    the condition that $\bar{P} \lambda^\ast = \lambda^\ast$ (note that $\bar{P}$
    is also a projection). $P$ will commute with $\mcB_1$ if and only $\bar{P}$
    commutes with $\bar{\mcB}_1$. Thus we arrive at the requirement that
    $\cl \bar{\mcB}_1 \lambda^\ast H_2 = H_1$.
\end{proof}
From this lemma it is easy to see that every strategy projects down to a unique
non-degenerate strategy. One more special case is important. The kernel of
$\lambda$ is the orthogonal complement of the image of $\lambda^\ast$, so $\cl
\lambda H_1 = \cl \rho H_2$. If $\rho$ commutes with $\mcB_2$ then $\rho H_2$
is an invariant subspace of $\mcB_2$, so $\cl \mcB_2 \lambda H_1 = H_2$ if and
only if $\cl \rho H_2 = H_2$.
\begin{proof}[Proof of Theorem \ref{T:tensordescript}] Suppose we are given
    an optimal non-degenerate strategy $\{A_i\}$, $\{B_j\}$, $\ket{\psi}$ for
    $G$, acting on Hilbert spaces $H_1$ and $H_2$. We fix a basis for $H_1$,
    and use the notation introduced above for $\ket{\psi}$. Define $u_i = (A_i
    \otimes I) \ket{\psi}$ and $v_j = (I \otimes B_j) \ket{\psi}$. The vectors
    $u_1,\ldots,u_m$ and $v_1,\ldots,v_n$ form an optimal vector strategy for
    $G$. Let $c_i$ denote the $i$th marginal row bias for $G$, and $d_j$ the
    $j$th marginal column bias. By Corollary \ref{C:mbias}, we know that
    \begin{equation*}
        d_j (I \otimes B_j) \ket{\psi} = \sum_i G_{ij} (A_i \otimes I) \ket{\psi},
    \end{equation*}
    where each $d_j$ is non-zero.  It follows from Lemma \ref{L:tensorcommute}
    that $\rho$ commutes with all $B_j$. As previously mentioned, Lemma \ref{L:nondegen}
    now implies that the closure of the image of $\rho$ is equal to $H_2$.

    Now we also know that
    \begin{equation*}
        c_i (A_i \otimes I) \ket{\psi} = \sum_j G_{ij} (I \otimes B_j) \ket{\psi},
    \end{equation*}
    and consequently by Lemma \ref{L:tensorcommute} $\sum_j G_{ij} B_j \lambda
    = c_i \lambda \bar{A_i}$. We apply the map $S \mapsto S S^\ast$ to both
    sides to get
    \begin{equation*}
        \left( \sum_j G_{ij} B_j\right)^2 \rho = \left(\sum_j G_{ij} B_j\right) \
            \rho\ \left(\sum_j G_{ij} B_j\right) = c_i^2 \lambda \bar{A_i}^2
            \lambda^\ast = c_i^2 \rho.  
    \end{equation*}
    Since the closure of the image of $\rho$ is equal to $H_2$, we find that
    $\left( \sum_j G_{ij} B_j\right)^2 = c_i^2 I$. Thus the operators
    $\{B_j\},\rho$ describe a cyclic representation of $\mcA$, and the density
    operator $\rho$ commutes with $\mcA$.

    Conversely, suppose we are given a cyclic representation $\{B_j\},\rho$ of
    $\mcA$, acting on a Hilbert space $H$, such that $\rho$ commutes with all
    $B_i$ and $\cl \mcA \rho H = H$. We can conclude that $\cl \rho H = H$, so
    if $\lambda$ is the square root of $\rho$, then $\rho H \subset \lambda H$,
    and hence $\cl \lambda H = H$. Now set $H_1 = H_2 = H$ and $\ket{\psi} =
    \sum_i \ket{i} (\lambda \ket{i})$. Define $\bar{A_i} = \frac{1}{c_i} \sum_j
    G_{ij} B_j $, so that $A_i$ is a self-adjoint operator squaring to the
    identity. From Lemma \ref{L:tensorcommute} we can conclude that 
    \begin{equation*}
        c_i (A_i \otimes I) \ket{\psi} = \sum_j G_{ij} (I \otimes B_j) \ket{\psi},
    \end{equation*}
    and hence
    \begin{align*}
        \sum_{i,j} G_{ij} \bra{\psi} A_i \otimes B_j \ket{\psi} = \sum_i c_i
    \end{align*}
    is the optimal success bias for $G$. Thus $\{A_i\}$, $\{B_j\}$,
    $\ket{\psi}$ is an optimal quantum strategy for $G$. It is non-degenerate
    because $\cl \lambda H = H$ and $\lambda^\ast = \lambda$.
\end{proof}

Now we can give the proof of Corollary \ref{C:minentangl}.
\begin{proof}[Proof of Corollary \ref{C:minentangl}]
    Let $\{B_j\},\rho$ be a cyclic representation of $\mcA$ on a Hilbert space
    $H$ such that $\rho$ commutes with all $B_j$. Let $\alpha_i$ be the
    eigenvectors of $\rho$. By the spectral theorem for self-adjoint compact
    operators (see \cite{zimmer} for the statement of this theorem), $H$
    decomposes into a direct sum of finite-dimensional eigenspaces
    $E_{\alpha_i}$ of $\rho$.  If $v \in E_{\alpha_i}$ is an eigenvector, then
    $\rho B_j v = B_j \rho v = \alpha_i B_j v$, and hence $\mcA$ preserves the
    eigenspaces of $\rho$. Because each $E_{\alpha_i}$ is finite-dimensional,
    each decomposes further into a sum of irreducible representations of
    $\mcA$. Thus we can conclude that $H$ decomposes into a direct sum of
    finite-dimensional irreducible representations $H_k$.

    Now we can write $\rho$ as a sum of density operators which have orthogonal
    support. Namely, let $\rho_k = I / \tr(I)$, where $I$ denotes the identity
    on $H_k$. Then $\rho = \sum p_k \rho_k$, where the $p_k$'s are positive
    scalars summing to one. By Theorem 11.8 of~\cite{nielsenchuang}, $S(\rho) =
    \sum p_k S(\rho_k) + H(p_k)$, where $S$ denotes the von Neumann entropy and
    $H$ denotes the Shannon entropy.  $H(p_k) > 0$ unless there is exactly one
    invariant subspace.  Hence the minimum entanglement will be attained only
    by an irreducible representation.  The entanglement used by an irreducible
    representation is $S(\rho_k) = \log_2 (\dim H_k)$.
\end{proof}

\section{Quantum strategies: the approximate case}\label{S:approxrep}

The goal of this section is to prove Theorem \ref{T:approxrep}. We rely on the
following lemma, which is a concatenation of lemmas concerning eigenvalue gaps
due to Babai and Friedl.
\begin{lemma}[Babai and Friedl, 1991]\label{L:evalgap} Let $\rho$ be a density
    matrix on the finite dimensional Hilbert space $\C^d$ such that $\tau :=
    \norm{ \rho - \Id / d} > 0$.  Then there is an orthogonal decomposition
    $\C^d = W_1 \oplus W_2$ such that
    \begin{itemize}
        \item $W_1$ and $W_2$ are non-trivial invariant subspaces of $\rho$, 
        \item if $\alpha_1$ and $\alpha_2$ are eigenvalues of $\rho$
            with eigenvectors in $W_1$ and $W_2$ respectively, then
            $\alpha_1 \geq \alpha_2 + \tau / d$, and
        \item if $\rho$ approximately commutes with another matrix $S$,
            meaning that $\norm{\rho S - S \rho}_F \leq \eps$, then
            \begin{equation*}
                \norm{P_1 S P_1 + P_2 S P_2 - S}_F \leq \frac{\eps d}{\tau},
            \end{equation*}
            where $P_i$ is the orthogonal projection onto $W_i$.
    \end{itemize}
\end{lemma}
\begin{proof}
    Write the eigenvalues of $\rho$ in descending order as $\lambda_1 \geq
    \ldots \geq \lambda_d$. Since $\sum \lambda_i = 1$, we must have
    $\lambda_1 \geq 1/d \geq \lambda_n$. Let $\Delta = \max |\lambda_i -
    \lambda_{i+1}|$ be the largest eigenvalue gap. Then
    \begin{equation*}
        \tau = \max |\lambda_i - 1/d | \leq \lambda_1 - \lambda_d \leq d \Delta,
    \end{equation*}
    so $\Delta \geq \tau / d$. Then apply Lemma 2.11 and Remark 2.7 of
    \cite{babaifriedl:1991}.
\end{proof}

\begin{proof}[Proof of Theorem \ref{T:approxrep}]
    We want to find a projection $P$ on $H_2$ such that $\norm{P - (P B_j
    P)^2}$ and $\norm{c_i^2 P - \left(\sum_l G_{il} P B_l P \right)^2 }$
    are bounded above by $C d^{5/3} \eps^{1/12}$ for all $1 \leq i \leq m$ and
    $1 \leq j \leq n$. We start by bounding $\norm{c_i^2 P - \left(\sum_l
    G_{il} P B_l P \right)^2 }$. Let $C' = \sqrt{10} (m+n)^{1/4}$.

    Let $\ket{\psi} = \sum \ket{i} \lambda \ket{i}$ and $\rho = \lambda
    \lambda^\ast$ be the partial trace of $\ket{\psi}$ with respect to $H_1$.
    If $\{A_i\}$,$\{B_j\}$, $\ket{\psi}$ is $\eps$-optimal, then by Theorem \ref{T:mbias},
    $\norm{ \left( c_i A_i \otimes \Id - \sum_j G_{ij} \Id \otimes B_j\right)
    \ket{\psi} } \leq C' \eps^{1/4}$.  It follows from Lemma
    \ref{L:tensorcommute} that $\norm{ c_i \lambda \bar{A} - \sum_j G_{ij} B_j
    \lambda }_F \leq C' \eps^{1/4}$.  It follows that $\norm{ c_i^2 \rho -
    \sum_j G_{ij} B_j \rho \sum_j G_{ij} B_j } \leq 3 C' \eps^{1/4}$.  But
    Lemma \ref{L:tensorcommute} also tells us that $\norm{ \rho \sum_j G_{ij}
    B_j - \sum_j G_{ij} B_j \rho}_F \leq 2 C' \eps^{1/4}$ for all $j$, so
    combining this with the last inequality, we get
    \begin{equation}\label{E:closeineq}
        \norm{\left( c_i^2 - \left( \sum_j G_{ij} B_j\right)^2 \right) \rho } \leq 5 C' \eps^{1/4}.
    \end{equation}
    Inequality (\ref{E:closeineq}) is potentially much weaker than the inequality we want,
    since $\rho$ might have arbitrarily small eigenvalues. Lemma \ref{L:evalgap}
    suggests two ways to strengthen this inequality. The first is to replace
    $\rho$ with the maximally mixed density matrix: let $\tau = \norm{ \rho -
    \Id / d}$ and observe that 
    \begin{align*} 
        \norm{ c_i^2 - \left(\sum_j G_{ij} B_j\right)^2 } & \leq 
            d \norm{ \left(c_i^2 - \left( \sum_j G_{ij} B_j\right)^2 \right) \rho } \\
            & + d \norm{ c_i^2 - \left(\sum_j G_{ij} B_j\right)^2 } \norm{\rho - \Id/d} 
                \leq 5 C' d \eps^{1/4} + 2 d \tau.
    \end{align*}
    If $\tau = 0$ then we can take $P = \Id$ and we will be done. If
    $\tau > 0$ then we can apply Lemma \ref{L:evalgap} to decompose $H_2$ as
    $W_1 \oplus W_2$. Let $P_i$ denote the orthogonal projection onto $W_i$.
    Let $B_j' = P_1 B_j P_1 + P_2 B_j P_2$ and let $\rho_i = P_i \rho P_i$.
    Because $\sum_j G_{ij} B_j$ approximately commutes with $\rho$, we get
    $\norm{ \sum_j G_{ij} B_j - \sum_j G_{ij} B_j' } \leq 2 C' d
    \eps^{1/4} / \tau$. If we let $D_{kl} = \sum_j G_{ij} P_k B_j P_l$, then
    \begin{equation*}
        \sum_j G_{ij} B_j - \sum_j G_{ij} B_j' = \begin{pmatrix} 0 & D_{12} \\
            D_{21} & 0 \end{pmatrix}, 
    \end{equation*}
    so $\norm{D_{12}},\norm{D_{21}} \leq 2 C' d \eps^{1/4} / \tau$, and
    \begin{equation*}
        \norm{\left( c_{i}^2 P_1 - D_{11}^2 - D_{12} D_{21} \right) \rho_1 } 
            \leq 5 C' \eps^{1/4}.
    \end{equation*}
    But $\norm{D_{12} D_{21}} \leq 4 C' d^2 \eps^{1/2} / \tau^2$, so
    \begin{equation*}
        \norm{\left( c_i^2 - D_{11}^2 \right) \rho_1} \leq 5 C' \eps^{1/4} 
            + 4 (C')^2 d^2 \eps^{1/2} / \tau^2.
    \end{equation*}
    Since $\rho_1$ has minimum eigenvalue at least $\tau / d$, we can
    conclude that
    \begin{align*}
        \norm{ c_i^2 P_1 - \left(\sum_j G_{ij} P_1 B_j P_1\right)^2 } & \leq
        \norm{ \left( c_i^2 P_1 - \left(\sum_j G_{ij} P_1 B_j P_1\right)^2\right)
            \rho_1 } \norm{ (\rho_1)^{-1} } \\ 
        & \leq \frac{5 C' d \eps^{1/4}}{\tau} + \frac{4 (C')^2 d^3 \eps^{1/2}}{\tau^3}.
    \end{align*}

    This leaves us to compare the efficacy of two choices for the projection
    $P$ required by the proposition. If we choose $P = \Id$, then we have an
    upper bound on $\norm{ c_i^2 P - \left( \sum_j G_{ij} P B_j P\right)^2 }$
    given by $f_1(\tau) = 5 C' d \eps^{1/4} + 2 d \tau$, while if we choose $P
    = P_1$ we have an upper bound given by $f_2(\tau) = 5 C' d \eps^{1/4} /
    \tau + 4 (C')^2 d^3 \eps^{1/2} / \tau^3$.  Let $\tau_0 = \frac{3}{2}
    (C')^{1/2} d^{1/2} \eps^{1/8}$, and observe that $f_1(\tau_0)$ and
    $f_2(\tau_0)$ are both bounded above by $8 (C')^{1/2} d^{3/2} \eps^{1/8}$.
    Since $f_1$ is increasing, if $\tau \leq \tau_0$ then we can take $P = \Id$
    and get the inequality required by the proposition. If $\tau > \tau_0$,
    then $f_2$ is decreasing on the interval $(0,\infty)$, so we can take $P =
    P_1$ and get the inequality required by the proposition.

    So far we have ignored the inequality we are supposed to get on 
    $\norm{ P - (P B_j P)^2 }$. If we take $P = \Id$, then $P = (P B_j P)^2$.
    So suppose $\tau > \tau_0$ and we take $P = P_1$. Now by Theorem \ref{T:mbias}
    applied to the columns, $\norm{ \left(d_j \Id \otimes B_j - \sum_i G_{ij}
    A_i \otimes \Id \right) \ket{\psi} } \leq C' \eps^{1/4}$.  Hence $\norm{ B_j
    \rho - \rho B_j }_F \leq 2 C' \eps^{1/4} / d_j$. As before, Lemma \ref{L:evalgap}
    tells us that $\norm{P_1 B_j P_2},\norm{P_2 B_j P_1} \leq 2 C' d \eps^{1/4}
    / (d_j \tau)$, so
    \begin{equation*}
        \norm{ P_1 - \left(P_1 B_j P_1\right)^2 } \leq \frac{4 (C')^2 d^2 \eps^{1/2}}
            {d_j^2 \tau^2}.
    \end{equation*}
    Let $f_3(\tau) = 4 (C')^2 d^2 \eps^{1/2} / (d_j^2 \tau^2)$. Then $f_3$ is
    decreasing and
    \begin{equation*}
        f_3(\tau_0) = \frac{8  C' d \eps^{1/4}}{9 d_j^2} \leq (C')^{1/2} d \eps^{1/8},
    \end{equation*}
    where the last inequality comes from the hypothesis that $\eps \leq
    d_j^{16}/ (C')^4$.  So we get $f_3(\tau) \leq 8 (C')^{1/2} d^{3/2}
    \eps^{1/8}$ as long as $\tau > \tau_0$.
\end{proof}

\section{Stable relations for the Clifford algebra}\label{S:cliffstable}

We use the following stability result for amenable groups due to Kazhdan to
prove that the relations defining the Clifford algebra are stable. 
\begin{defn} An $\eps$-representation of a topological group $G$ is
    a continuous map $\phi$ from $G$ to the group of unitary transformations
    on a Hilbert space, such that
    \begin{equation*}
        \norm{ \phi(gg') - \phi(g) \phi(g') } \leq \eps \text{ for all } g,g' \in G.
    \end{equation*}
    Here $\norm{\cdot}$ is the operator norm.
\end{defn}
 
\begin{thm}[Kazhdan, Theorem 1  of \cite{kazhdan:1982}] Let $G$ be an amenable
    group and $\phi : G \arr U$ an $\eps$-representation of $G$, for
    some $0 < \eps < 1/100$. Then there is a representation $\pi : G \arr U$
    such that $\norm{ \phi(g) - \pi(g) } \leq \eps$ for all $g \in G$.
\end{thm}

Let $G_n$ be the multiplicative subgroup generated by the elements
$Y_1,\ldots,Y_n$ in the Clifford algebra $\mcC_n$. As a set, $G_n$ is equal to
\begin{equation*}
    \left\{ J^{a_0} Y_0^{a_1} \cdots Y_n^{a_n}\ :\ a_0,\ldots,a_n \in \{0,1\} \right\}
\end{equation*}
where we use $J$ to denote $-\Id$ for clarity. $G_n$ is a finite group of order
$2^{n+1}$, and hence is an amenable group with the discrete topology. 

\begin{lemma}\label{L:cliffstable} Suppose $B_1,\ldots,B_n$ is an
    $\eps$-approximate representation of $\mcC_n$ on the finite-dimensional
    Hilbert space $\C^d$, where $0 < \eps < 1/(250 n^2)$. Then there is a
    representation $B_1',\ldots,B_n'$ of $\mcC_n$ on $\C^d$ such that $\norm{B_i - B_i'}
    \leq 5 n^2 \eps / 2$. 
\end{lemma}
\begin{proof}
    Since $B_i$ is self-adjoint and $\norm{B_i^2 - \Id} \leq \eps$, every 
    eigenvalue of $B_i$ is within $\eps$ of either $+1$ or $-1$. Hence there
    is a self-adjoint unitary matrix $\hat{B}_i$ such that $\norm{B_i -
    \hat{B}_i} \leq \eps$. Since $\norm{B_i B_j + B_j B_i} \leq \eps$, it
    follows that $\norm{\hat{B}_i \hat{B}_j + \hat{B}_j \hat{B}_i} \leq 5
    \eps$. Suppose $J^{a_0} Y_1^{a_1} \cdots Y_n^{a_n} \cdot J^{b_0} Y_1^{b_1}
    \cdots Y_n^{b_n} = J^{c_0} Y_1^{c_1} \cdots Y_n^{c_n}$. This means that the
    product $J^{a_0} Y_1^{a_1} \cdots Y_n^{a_n} \cdot J^{b_0} Y_1^{b_1} \cdots
    Y_n^{b_n}$ can be transformed to $J^{c_0} Y_1^{c_1} \cdots Y_n^{c_n}$ using
    the relations $Y_i Y_j = - Y_j Y_i$ at most $n (n-1)/2$ times. So the map
    \begin{equation*}
        \psi : J^{a_0} Y_1^{a_1} \cdots Y_n^{a_n} \mapsto (-1)^{a_0} \hat{B}_1^{a_1}
            \cdots \hat{B}_n^{a_n}
    \end{equation*}
    gives a unitary $\eps'$-representation of $G_n$, where $\eps' = 5
    \binom{n}{2} \eps$. By Kazhdan's theorem, there is a representation
    $\phi$ of $G_n$ with $\norm{\phi - \psi} \leq \eps'$ and 
    $\norm{\phi(Y_i) - B_i} \leq \eps' + \eps \leq 5 n^2 \eps / 2$.
    Let $B_i' = \phi(Y_i)$. Now $B_i' B_j' = \phi(J) B_j' B_i'$, so we will
    be done if we can show that $\phi(J) = -\Id$. But $\psi(J) = -\Id$, so
    $\norm{ \psi(J) + \Id} \leq \eps' < 1$. Since $\psi(J)$ is unitary and
    squares to $\Id$, it follows that $\psi(J) = -\Id$.
\end{proof}

\begin{proof}[Proof of Proposition \ref{P:cliffstable}]
    Let $E^{kl}$ denote the $n \times n$ matrix with ones in the $kl$th
    and $lk$th positions, and zeroes everywhere else. Map every $m \times n$
    XOR non-local game $G$ to an $m$-tuple of rank one $n \times n$
    matrices $G^1,\ldots,G^m$, where $G^i_{kl} = G_{ik} G_{il}$. If the $i$th
    row of $G$ is equal to the vector with ones in the $l$th and
    $k$th positions and zeroes elsewhere, then $G^i$ will be equal to 
    $E^{kk} + E^{ll} + E^{kl}$. Thus whenever $m \geq \binom{n}{2}$, there
    is at least one XOR non-local game such that the matrices $G^1,\ldots,G^m$,
    $E^{11},\ldots,E^{nn}$ span the space of $n\times n$ symmetric matrices. 
    We can think of the space of $m \times n$ XOR non-local games either
    as $\R^{mn}$ (without normalization) or as the $mn-1$-sphere in $\R^{mn}$
    (with normalization). In either case, the subset of XOR non-local games
    for which the matrices $G^1,\ldots,G^m,E^{11},\ldots,E^{nn}$ are linearly
    dependent is an algebraic set. Hence for almost all $m \times n$ games the
    matrices $G^1,\ldots,G^m$,$E^{11}, \ldots,E^{nn}$ span the space of
    symmetric matrices. Assume that $G$ is a game for which this is
    true.

    Let $M$ be the $n \times n$ matrix with coefficients $M_{ij} = X_i X_j$ in
    $\mcA$. The defining relations for $\mcA$ can be rewritten in this
    notation as 
    \begin{equation*} 
        E^{jj} \cdot M = \Id_{\mcA} \text{ and } G^i \cdot M = c_i^2 \cdot
            \Id_{\mcA},
    \end{equation*}
    where $X \cdot Y$ denotes the coordinate-wise $\mcA$-valued scalar product.
    For every $1 \leq k,l \leq n$ there are scalars $s_{j}$ and $t_i$ such that
    \begin{equation*}
        \sum_j s_j E^{jj} + \sum t_i G^i = E^{kl}.
    \end{equation*}
    Let $V_{kl} = \sum_i t_i c_i^2 + \sum_j s_j$. Then
    \begin{align}\label{E:strongcliffrel}
        X_k X_l + X_l X_k - V_{kl} \Id_{\mcA} & = E^{kl} \cdot M - V_{kl} \\
            & = \sum_j s_j (E^{jj} \cdot M - \Id_{\mcA}) + \sum_i t_i (G^i \cdot M
                - c_{i}^2 \Id_{\mcA}).
    \end{align}
    So for almost all $m \times n$ games the solution algebra is strongly Clifford.

    Now assume that $\mcA$ is strongly Clifford. There are self-adjoint
    elements $Y_1,\ldots,Y_r \in \mcA$ satisfying the usual Clifford relations,
    where $r = \rank V$. Each $Y_i$ is a linear combination $\sum y_{ij} X_j$
    of the $X_j$'s, and $2 \delta_{ik} = Y_i Y_k + Y_k Y_i = \sum_{j,l} y_{ij}
    y_{kl} (X_j X_l + X_l X_j) = \sum_{j,l} y_{ij} y_{kl} V_{jl}$. Suppose
    $(B_1,\ldots,B_n)$ is an $\eps$-approximate representation of $\mcA$, and
    let $\tilde{Y}_i = \sum y_{ij} B_j$. Then
    \begin{equation*}
        \tilde{Y}_i \tilde{Y}_k + \tilde{Y}_k \tilde{Y}_i - 2 \delta{ij}
            = \sum_{j,l} y_{ij} y_{kl} (B_j B_l + B_l B_j - V_{jl} \Id).
    \end{equation*}
    The relations $X_j X_l + X_l X_j - V_{jl} \Id$ are in turn given by
    equation (\ref{E:strongcliffrel}) for some choice of constants $s_j$
    and $t_i$. It follows that $\norm{\tilde{Y}_i \tilde{Y_k} + \tilde{Y}_k
    \tilde{Y}_i - 2\delta_{ij} } \leq D' \eps$ for some (possibly large)
    constant $D'$. In particular $\norm{\tilde{Y}_i} \leq 1 + D' \eps
    / 4$. If we assume $D' \eps \leq 1$ then we can truncate eigenvalues
    of $\tilde{Y}_i$ to $1$ to get a $5 D' \eps / 2$-approximate
    representation of $C_r$. Take $\delta = 1 / (625 D' r^2)$. If
    $\eps < \delta$ then $\eps D' < 1$ and $5 D' \eps / 2 < 1/ (250 r^2)$, so
    we can apply the above lemma to get an exact representation.
\end{proof}

\section{Examples}\label{S:examples}

\subsection{XOR non-local games constructed from a graph}

There is a general method which allows us to construct many examples of games
requiring high entanglement. Start with a graph $G$ with $v$ vertices and $e$
edges. We construct a matrix $A$ which has two rows for each edge of $G$, and
columns indexed by the vertices. If $ij$ is an edge in $G$ with $i<j$, then the
first row corresponding to $ij$ will contain a $1$ in the $i$th column, a $-1$
in the $j$th column, and zeroes everywhere else. The other row will contain a
$1$ in both the $i$th and the $j$th column, with zeroes everywhere else. Let
$A_G$ be the XOR non-local game with game matrix $A/4e$.

For example, if $K_2$ is the complete graph on $2$ vertices, and $G$ is the
graph with $3$ vertices and edge set $\{12, 13\}$, then
\begin{equation*}
    A_{K_2} = \frac{1}{4}
            \begin{pmatrix} 1 & -1 \\
                            1 &  1 
            \end{pmatrix}
    \text{ and }
    A_G = \frac{1}{8} 
          \begin{pmatrix} 1 & -1 &  0 \\
                          1 &  1 &  0 \\
                          1 &  0 & -1 \\
                          1 &  0 &  1  
           \end{pmatrix}
\end{equation*}

\begin{prop}\label{P:graphgame} Let $A_G$ be the $2e \times v$ game constructed
    from the graph $G$ according to the above prescription. Then the optimal
    solutions to the associated convex programming problem $(\Gamma)$ of
    Proposition \ref{C:nonlinear} are the semidefinite matrices $V$ satisfying the
    conditions 
    \begin{equation*}
        V_{ii} = 1 \text{ for all } 1 \leq i \leq v \text{ and }
        V_{ij} = 0 \text{ for each edge } ij \text{ of } G.
    \end{equation*}
    The optimal quantum bias for $A_G$ will be $1/2$ and the optimal 
    classical bias will be $1/\sqrt{2}$. The solution algebra for $A_G$ will be
    the $C^*$-algebra generated by $X_1, \ldots,X_v$ satisfying the relations 
    \begin{equation*}
        X_i^2 = \Id \text{ for all } 1 \leq i \leq v \text { and }
        X_i X_j = -X_j X_i \text{ for each edge } ij \text{ of } G. 
    \end{equation*}
\end{prop}
\begin{proof}[Proof of Proposition \ref{P:graphgame}]
    For convenience, let $f(t)$ denote the function $\sqrt{1+t} + \sqrt{1-t}$.
    The optimization problem $(\Gamma)$ has objective function
    \begin{equation*}
        \sum_{ij \in E(G)} \frac{1}{2 \sqrt{2} e} f(V_{ij}),
    \end{equation*} 
    using the fact that $V_{ii} = 1$ at any feasible point.  Now $f(0) = 2$,
    and this is the unique maximum of $f(t)$ on the interval $[-1,1]$. If $V$
    is feasible for $(\Gamma)$ then $V_{ij} \in [-1,1]$, so $(\Gamma)$ has optimum
    value $1/\sqrt{2}$, and the optimal solutions are the feasible solutions $V$
    with $V_{ij} = 0$ if $ij \in E(G)$. 

    On the other hand, if $V$ corresponds to a classical solution, then $V_{ij}
    = \pm 1$, and $f(\pm 1) = \sqrt{2}$. 

    If $V$ is optimal for $(\Gamma)$ then the squares of the two row biases
    corresponding to edge $ij$ are 
    \begin{equation*}
        c = \frac{\sqrt{(V_{ij}^2 \pm 2 V_{ij}^2 + V_{jj}^2)}}{4 e} =
            \frac{1}{2\sqrt{2} e}.
    \end{equation*}
    Given that $X_i^2 = X_j^2 = \Id$, the relations
    \begin{equation*}
        \left(\frac{X_i}{4 e} \pm \frac{X_j}{4 e}\right)^2 = c^2 \Id
    \end{equation*}
    are equivalent to the single relation $X_i X_j = - X_j X_i$.
\end{proof}
The game described by $A_{K_2}$ is commonly known as the $\CHSH$ game. As an
obvious generalization of the CHSH game, we define $\CHSH(n)$ to be the XOR
game with game matrix $A_{K_n}$. The associated solution algebra is
strongly Clifford of rank $n$, and every irreducible representation of the
solution algebra has dimension $2^{\lfloor n/2\rfloor}$. Thus Corollaries
\ref{C:minentangl} and \ref{C:dimbound} give lower bounds on the entanglement
required for optimal and near-optimal strategies. For $\CHSH(n)$ we also can
derive a lower bound on the rank of near-optimal vector strategies. 

\begin{prop}\label{P:vectlb} Let $\eps > 0$.  Every $\eps$-optimal vector
    strategy for $\CHSH(n)$ is supported on a Hilbert space of dimension at
    least $n - 8 \sqrt{2} n (n-1) \eps$.
\end{prop}
\begin{proof}
    For the first part of the Proposition, let $\{u_{ijl}\}_{l \in \{0,1\}}$,
    $\{v_j\}$ be an $\eps$-optimal vector strategy. Let $V$ be the matrix with
    $V_{ij} = v_i \cdot v_j$. Observe that
    \begin{equation*}
        \sum_{1 \leq i<j \leq n} \frac{1}{\sqrt{2} n(n-1)} f(V_{ij}) 
            \geq \sum_{1 \leq i<j \leq n} \frac{1}{2 n(n-1)} u_{ij0} (v_i - v_j)
                + \frac{1}{2 n(n-1)}  u_{ij1} (v_i + v_j)
            \geq \frac{1}{\sqrt{2}}- \eps,
    \end{equation*}
    where $f(t) = \sqrt{1+t} + \sqrt{1-t}$ again.  Using the Taylor series
    expansion, we see that $f(t) \leq 2 - t^2 / 4$ for all $t \in [-1,1]$.
    Consequently 
    \begin{equation*}
        \sum_{1 \leq i<j \leq n} \frac{1}{\sqrt{2} n(n-1)} f(V_{ij})
        \leq \frac{1}{\sqrt{2}} - \frac{\norm{V - I}_F^2}{8\sqrt{2} n(n-1)}.
    \end{equation*}
    Now $\dim \vspan_{\R} \{v_j\} = \rank V = n-z$, where $z$ is the number
    of eigenvalues of $V-I$ equal to $-1$. But $\norm{V - I}_F^2 \geq z$,
    implying that 
    \begin{equation*}
        n-z \geq n - \norm{V-I}_F^2 \geq n - 8\sqrt{2} n (n-1) \eps.
    \end{equation*}
\end{proof}

\subsection{A game for which Tsirelson's bound is tight}

Because $\CHSH(n)$ is an $n(n-1) \times n$ game, $\CHSH(n)$ does not quite meet
Tsirelson's upper bound on entanglement.  To give a game which does meet
Tsirelson's upper bound (for $n \geq 2$), let $m = n(n-1)/2$. Define a matrix
$G$, with rows indexed by pairs $ij$, where $1 \leq i < j \leq n$, and columns
indexed by the set $\{1,\ldots,n\}$, as 
\begin{equation}\label{ourgame}
    G_{ij,k} = \begin{cases} 1/2m & i=k \\
                             -1/2m & j=k \\
                            0 & \text{ otherwise }
                \end{cases}
\end{equation}
Then $G$ defines an $m \times n$ XOR non-local game. Now look at the optimization
problem $(\Gamma)$ for this game. The objective function of $(\Gamma)$ is invariant
under conjugation of $V$ by permutation matrices. Thus we can solve $(\Gamma)$ by
the standard technique of summing over the action of the symmetric group to reduce
to a linear program. It follows that $V = (2 m^2 s^2) I + (1-2 m^2 s^2) E$
is a solution to $(\Gamma)$, where $E$ is the matrix of all ones. The marginal
row biases can be determined by the formula $c_{ij}^2 = (1- V_{ij}) / 2m^2$.
The solution algebra will be a Clifford algebra of rank $n-1$, and thus $G$
requires $\lfloor (n-1)/2\rfloor$ ebits.

\subsection{Clifford strategies do not always use minimal entanglement}

To give an example of a game for which the strategy of minimum entanglement
is not Clifford, we use the graph construction of Proposition \ref{P:graphgame}.
Given an integer $n$, let $Y_1,\ldots,Y_n$ be the canonical generators of the
Clifford algebra $\mcC_n$.  Define a graph $G$ whose vertices are ordered
$k$-tuples $(i_1,\ldots,i_k)$, where $1 \leq i_1 < \ldots < i_k \leq n$, and
$k$ ranges between $1$ and $n$. Place an edge in $G$ between $(i_1,\ldots,i_k)$
and $(j_1,\ldots,j_l)$ if $Y_{i_1} \cdots Y_{i_k}$ anticommutes with $Y_{j_1}
\cdots Y_{j_l}$ in $\mcC_n$. Let $CL(n)$ be the XOR non-local game $A_G$
constructed from the graph $G$ according to the recipe in Proposition
\ref{P:graphgame}. We want to show that the optimal strategies of minimum
entanglement for $CL(n)$ are not Clifford.
\begin{prop}\label{P:highrank}
    Let $(\Gamma)$ be the convex programming problem associated to $CL(n)$ by
    Proposition \ref{C:nonlinear}. For $n \geq 3$, every solution $V$ to $(\Gamma)$
    has rank at least $(4n-8)/3$. 
\end{prop}
\begin{proof}
    We use a pigeon-hole argument. Let $V$ be a solution to $(\Gamma)$, and take
    unit vectors $v(\tilde{i})$ in $\R^N$ such that $V_{\tilde{i},\tilde{j}} =
    v(\tilde{i}) \cdot v(\tilde{j})$. For convenience of notation we ignore order
    in the multi-indices, so for example $v(2,1)$ would be the same as $v(1,2)$.
    Now $ij$ is an edge in $G$ for every $i \neq j$, and hence $v(i) \cdot
    v(j) = 0$. Let $W = \vspan \{v(i) : 1 \leq i \leq n\}$, and write
    \begin{equation*}
        \vspan \{v(\tilde{i}) \} = W \oplus W^\perp.
    \end{equation*} 
    It will also be convenient to let 
    \begin{equation*}
        W(i_1,\ldots,i_k) = \vspan \{v(i_1),\ldots,v(i_k)\} \oplus W^\perp.
    \end{equation*}

    Now there are $\binom{n}{3}$ vectors of the form $v(ijk)$. If $l \not\in
    \{i,j,k\}$ then $v(l) \cdot v(ijk) = 0$. Hence $v(ijk)$ is contained in
    $W(ijk)$, and has non-zero projection onto one of $W(ij)$, $W(ik)$, or
    $W(jk)$. For $i<j$ let $Z_{ij}$ be the set of vectors $v(ijk)$ which
    have a non-zero projection onto $W(ij)$ (note that $k$ need not be greater
    than $i$ and $j$). There are $n-2$ vectors which could potentially be in
    $Z_{ij}$. However if $k \neq l$ then
    \begin{equation*}
        \left(P_{W(ij)} v(ijk)\right)\cdot\left( P_{W(ij)} v(ijl)\right)
            = v(ijk) \cdot v(ijl) = 0,
    \end{equation*}
    where $P_{W(ij)}$ is the orthogonal projection onto $W(ij)$. Thus $Z_{ij}$
    contains no more than $\dim W(ij) = 2 + \dim W^\perp$ vectors. Thus we
    conclude that
    \begin{equation*}
    \binom{n}{3} \leq \left(2 + \dim W^\perp\right) \binom{n}{2}.
    \end{equation*}
    Solving this inequality we find that $\dim W^\perp \geq \frac{n-8}{3}$,
    and hence
    \begin{equation*}
        \dim \vspan \{v(\tilde{i})\} \geq n + \frac{n-8}{3} = \frac{4n-8}{3}
    \end{equation*}
\end{proof}
Let $\mcA$ be the solution algebra associated to $CL(6t+2)$.  Every optimal
vector solution to $CL(6t+2)$ has rank at least $8t$, so every Clifford
representation uses at least $4t$ ebits. However, it is clear that any
irreducible representation of $\mcC_{6t+2}$ is also a representation of $\mcA$,
since we can send $X_{i_1,\ldots,i_k}$ to $i^s Y_{i_1} \cdots Y_{i_k}$, where
$s = \binom{k}{2}$. Thus $CL(6t+2)$ requires at most $3t+1$ ebits. It follows
that any Clifford representation uses at least $t-1$ ebits more than what is
required. 

\section{Acknowledgments}

I thank Richard Cleve, Alex Fink, and Sarvagya Upadhyay for helpful
discussions. I also thank Richard Cleve, Andrew Marks, and Sarvagya Upadhyay
for comments on the manuscript. This work was supported in part by NSERC. 

\bibliographystyle{plain}

\end{document}